\documentclass[11pt]{article}

\usepackage{amsmath, amsthm, amssymb, graphicx, enumerate, fullpage, url}
\usepackage{color, hyperref}
\usepackage{algorithm, algpseudocode}

\hypersetup{
	colorlinks=true,
	linkcolor=blue,
	citecolor=magenta
}

\theoremstyle{plain}% default
\newtheorem{theorem}{Theorem}[section]
\newtheorem{proposition}[theorem]{Proposition}
\newtheorem{lemma}[theorem]{Lemma}

\newtheorem{claim}[theorem]{Claim}

\theoremstyle{definition}
\newtheorem{definition}[theorem]{Definition}

\newcommand{\F}{\mathbb{F}}

\newcommand{\Z}{\mathbb{Z}}

\newcommand{\cA}{\mathcal{A}}
\newcommand{\cC}{\mathcal{C}}
\newcommand{\cL}{\mathcal{L}}

\newcommand{\tl}{\triangleleft}

\DeclareMathOperator{\agr}{agr}
\DeclareMathOperator{\aHom}{aHom}

\DeclareMathOperator{\Eq}{Eq}
\DeclareMathOperator{\Hom}{Hom}

\DeclareMathOperator{\poly}{poly}

\begin{document}
%%%%%%%%%%%%%%%%%%%%%%%%%%%%%%%%%%%%%%%%%%%%%%%%%%%%%%%%%%%%%%%%%%%%%%%%%%%%%%%
%%%%%%%%%%%%%%%%%%%%%%%%%%%%%%%%%%%%%%%%%%%%%%%%%%%%%%%%%%%%%%%%%%%%%%%%%%%%%%%

\title{List decoding group homomorphisms between supersolvable groups}
\author{Alan Guo%
\thanks{CSAIL, Massachusetts Institute of
Technology, 32 Vassar Street, Cambridge, MA, USA. {\tt aguo@mit.edu}. Research
supported in part by NSF grants CCF-0829672, CCF-1065125,
and CCF-6922462, and an NSF Graduate Research Fellowship}
\and Madhu Sudan%
\thanks{Microsoft Research, 1 Memorial Drive, Cambridge, MA 02142. {\tt
madhu@mit.edu}.}
}
\maketitle

%%%%%%%%%%%%%%%%%%%%%%%%%%%%%%%%%%%%%%%%%%%%%%%%%%%%%%%%%%%%%%%%%%%%%%%%%%%%%%%
\begin{abstract}
%%%%%%%%%%%%%%%%%%%%%%%%%%%%%%%%%%%%%%%%%%%%%%%%%%%%%%%%%%%%%%%%%%%%%%%%%%%%%%%
We show that the set of homomorphisms between two supersolvable groups can be
locally list decoded up to the minimum distance of the code, extending the
results of Dinur et al who studied the case where the groups are abelian.
Moreover, when specialized to the abelian case, our proof is more streamlined
and gives a better constant in the exponent of the list size. The constant
is improved from about 3.5 million to 105.

%%%%%%%%%%%%%%%%%%%%%%%%%%%%%%%%%%%%%%%%%%%%%%%%%%%%%%%%%%%%%%%%%%%%%%%%%%%%%%%
\end{abstract}
%%%%%%%%%%%%%%%%%%%%%%%%%%%%%%%%%%%%%%%%%%%%%%%%%%%%%%%%%%%%%%%%%%%%%%%%%%%%%%%

%%%%%%%%%%%%%%%%%%%%%%%%%%%%%%%%%%%%%%%%%%%%%%%%%%%%%%%%%%%%%%%%%%%%%%%%%%%%%%%
%%%%%%%%%%%%%%%%%%%%%%%%%%%%%%%%%%%%%%%%%%%%%%%%%%%%%%%%%%%%%%%%%%%%%%%%%%%%%%%
\section{Introduction}
\label{section:introduction}
%%%%%%%%%%%%%%%%%%%%%%%%%%%%%%%%%%%%%%%%%%%%%%%%%%%%%%%%%%%%%%%%%%%%%%%%%%%%%%%
%%%%%%%%%%%%%%%%%%%%%%%%%%%%%%%%%%%%%%%%%%%%%%%%%%%%%%%%%%%%%%%%%%%%%%%%%%%%%%%

It is well-known that for any pair of groups $G$ and $H$ with $G$ being
finite, the set of homomorphisms from $G$ to $H$ form an error-correcting
code. The most classical example of such a setting is when $G$ is the
additive group over $\F_2^n$ and $H = \F_2$ (where $\F_q$ denotes the
finite field of size $q$). The seminal work of Goldreich and Levin~\cite{GL89}
gave an ``efficient local list-decoding'' algorithm for this particular
setting. Such an algorithm has oracle access to
a function $f:\F_2^n \to \F_2$, and
given $\epsilon > 0$, reports all homomorphisms $\phi$ that agree with
$f$ on $1/2 + \epsilon$ fraction of the points in time $\poly(\log |G|,
\log |H|, 1/\epsilon)$.

A natural question, given the centrality of the Goldreich-Levin algorithm in
coding theory and learning theory, is to ask what is the most general
setting in which it works.
In particular, one abtraction of the (original) Goldreich-Levin algorithm is that it uses coding theory (in particular, the Johnson bound of coding theory) to get a combinatorial bound on the list size, namely the number of functions that may have agreement $1/2 + \epsilon$ with the function $f$.
It then uses some decomposability properties of the domain $\F_2^n$ to get
an algorithm for the list-decoding.
Grigorescu et al.~\cite{GKS06}
and Dinur et al.~\cite{DGKS08}, extended this abstraction to
the more general setting of abelian groups.
They first analyze $\delta_{G,H}$, the minimum possible
distance between two homomorphisms from $G$ to $H$.
They then consider the task of recovering all homomorphisms
at distance $\delta_{G,H} - \epsilon$ from a given function $f$.
Roughly they show
that the ``decomposability'' used in the algorithmic step
of Goldreich and Levin can be
generalized to the case of direct products of groups: so if $G =
G_1 \times G_2 \times \cdots \times G_k$ and each $G_i$ is small
and also if $H$ is small, then algorithmic step can be extended.
This reduces
the list-decoding question to the combinatorial one. Here the
standard bounds from coding theory are insufficient, however one
can use decompositions of the group $H$ into prime cyclic groups
to show that the list size is at most $\poly(1/\epsilon)$.

In this work, we take this line of work a step further and explore this
algorithm in the setting where $G$ and $H$ are not abelian. In this setting
decompositions of $G$ and $H$ turn out to be more complex, and indeed even
the question of determining $\delta_{G,H}$ turns out to be non-trivial.
This question is explored in a companion work by the first
author~\cite{G14}, where $\delta_{G,H}$ is determined explicitly for
a broad class of groups, including the case of ``supersolvable'' groups
which we study here. To describe the groups we consider
we recall some basic
group-theoretic terminology.

A subset $N \subseteq G$ is a subgroup of $G$, denote $N \leq G$, if
$N$ is closed under the group operation.
A subgroup $N \leq G$ is said to be normal in $G$,
denoted $N \tl G$,
if $aN = Na$ for all $a \in G$,
where $aN = \{an | n \in N \}$ and $Na = \{na | n \in N\}$.
If $N \tl G$, then the set of cosets of $N$ in $G$ form a group
under the operation $(aN)(bN) = (abN)$. This group is denoted
$G/N$.
$G$ is {\em solvable} if there
exists a series of groups
$\{1_G\} = G_0 \tl G_1 \tl \cdots \tl G_k = G$ such
that $G_i/G_{i-1}$ is abelian for every $i$.
We refer to the sequence
$\langle 1_G = G_0,G_1,\ldots,G_k = G\rangle$ as the
solvability chain of $G$.
$G$ is \emph{supersolvable} if it has a solvability chain
$\langle 1_G = G_0,G_1,\ldots,G_k = G\rangle$ where
$G_i \tl G$ and $G_i/G_{i-1}$ is cyclic for every $i$.

%%%%%%%%%%%%%%%%%%%%%%%%%%%%%%%%%%%%%%%%%%%%%%%%%%%%%%%%%%%%%%%%%%%%%%%%%%%%%%%
\subsection{Our results}
%%%%%%%%%%%%%%%%%%%%%%%%%%%%%%%%%%%%%%%%%%%%%%%%%%%%%%%%%%%%%%%%%%%%%%%%%%%%%%%

Our main results, stated somewhat informally, are the following:
\begin{itemize}
\item%
\textbf{(Combinatorial list decodability)}
There exists a constant $C \approx 105$ such that if $G$ and $H$ are
{\em supersolvable} groups, then for any $f:G \to H$, the number of (affine)
homomorphisms from $G$ to $H$ disagreeing with $f$ on less than
$\delta_{G,H} -\epsilon$ fraction of $G$ is at most $(1/\epsilon)^C$.
\item%
\textbf{(Algorithmic list decodability)}
Let $G$ be a solvable group and $H$ be any group such that the set of homomorphisms
from $G$ to $H$ have nice combinatorial list-decodability, i.e., the
number of homomorphisms from $G$ to $H$ that have distance
$\delta_{G,H} - \epsilon$ from a fixed function $f$ is at most
$(1/\epsilon)^C$. Then, the set of homomorphisms from
$G$ to $H$ can be locally list decoded up to $\delta_{G,H}-\epsilon$ errors
in $\poly(\log|G|,\log|H|,\frac1\epsilon)$ time assuming oracle access to
the multiplication table of $H$.\footnote{For the group $G$ we only need to
be able sample its elements in a specific way, and compute $f$ on elements
sampled in such a way. Using the (super)solvability of $G$, we can guarantee
that such sampling oracle of size $\poly \log |G|$ can be provided for every
$G$. For $H$ we are not aware of a similar result which allows for a
presentation of its elements, and providing access to the group
operation with size $\poly \log |H|$. Hence we are forced to make this an
explicit assumption.}
\end{itemize}

Putting the two ingredients together we get efficient list-decoding
algorithms up to radius $\delta_{G,H} - \epsilon$ whenever $G$ and $H$ are
supersolvable.

\subsection{Motivation and Contributions}

The study of list-decoding of homomorphisms is motivated by a few
objectives. First, an abstraction of the list-decoding algorithm highlights
the minimal assumptions needed to make it work. Here our work extends the
understanding in terms of reducing the dependence on commutativity (and so
in principle can apply to the decoding of matrix-valued functions).

A second motivation, emerging from
the works of \cite{GKS06,DGKS08}, is to extend combinatorial analyses
of list-decoding to settings beyond those where the Johnson bound is
applicable. Specifically the previous works used the Johnson bound when
the target group was $\Z_p$ for prime $p$ and then used the group-theoretic
framework to extend the analysis first to the case of cyclic groups of prime
power (so $H = \Z_{p^k}$ for prime $p$ and integer $k$) and then to the case
of general abelian groups. Each one of these steps lost in the exponent.
Specifically \cite{DGKS08} gave a function $C$ mapping constants
to constants such that the list size grew as $(1/\epsilon)^{C(2)}$ when
$H = \Z_{p^k}$ and $(1\epsilon)^{C(C(2))}$ for general groups.
They didn't calculate the exponents explicitly, but $C(2) \approx 105$
and $C(C(2)) \approx 3.5 \times 10^6$.
Our more general abstraction ends up cleaning up their proof significantly,
and even improve their exponent significantly. Specifically, we are able to
apply the inductive analysis implicit in previous works directly to the
solvability chain of $H$ (rather than working with the product structure)
and this allows us to merge the two steps in previous works to get
a list-size bound of $(1/\epsilon)^{C(2)}$ for all supersolvable groups.
Thus the abstraction and generalization improves the list-size bounds
even in the abelian case.
Our analysis shows that the list-decoding radius is as large as the
distance. We note that
there are relatively few cases of
codes that are known to be list-decodable up to their
minimum distance.
This property is shown to be true for
folded Reed-Solomon codes~\cite{GR08,Gur11},
derivative/multiplicity codes~\cite{GW11,Kop12},
Reed-Muller codes~\cite{GKZ08,Gop13},
and homomorphisms between abelian groups~\cite{GKS06,DGKS08}.

Finally, a potential objective would be to
get new codes with better list-decodability than existing codes.
Unfortunately, this hope remains unrealized in this work as well
as in \cite{GKS06,DGKS08}.

\subsection{Overview of proof}
%%%%%%%%%%%%%%%%%%%%%%%%%%%%%%%%%%%%%%%%%%%%%%%%%%%%%%%%%%%%%%%%%%%%%%%%%%%%%%%

We first prove the combinatorial bound on the list size by following the
framework developed by~\cite{DGKS08}, which works as follows.
First, find groups $\{1\}=H_{(0)},H_{(1)},\ldots,H_{(m)}=H$ in such a way
that any homomorphism $\phi \in \Hom(G,H)$ naturally induces a homomorphism
$\phi^{(i)} \in \Hom(G,H_i)$. This gives a natural notion of ``extending''
a homomorphism $\psi \in \Hom(G,H_i)$: $\phi$ extends $\psi$ if
$\phi^{(i)} = \psi$.
One then shows inductively that if
$\psi \in \Hom(G,H_i)$ has significant agreement with $f^{(i)}$, then there are
not too many $\phi \in \Hom(G,H)$ extending $\psi$ with significant agreement
with~$f$.
In~\cite{DGKS08}, $H$ is abelian and is decomposed as
$H = \Z_{{p_1}^{r_1}}^{e_1} \oplus \cdots \oplus \Z_{{p_r}^{r_m}}^{e_m}$.
One may take $H_{(i)}$ to be the direct sum of all but the last $i$ summands.
Then every $f:G \to H$ is naturally written as $f = (f_1,\ldots,f_m)$ where
$f_i:G \to \Z_{{p_i}^{r_i}}^{e_i}$, and thus $f^{(i)} = (f_1,\ldots,f_{m-i})$.
Now, to show the inductive claim for $H$, they reduce to the special cases where
$H = \Z_p^r$ and where $H = \Z_{p^r}$, and go through the same approach for
the special cases too. This goes through the ``special intersecting family''
theorem of~\cite{DGKS08} twice, resulting in a huge blowup in the exponent
of the list size.
Our proof differs from that of~\cite{DGKS08} as we prove the full inductive
claim directly, without reducing to any special cases, resulting in a much
smaller exponent.
However, for technical reasons, we only manage to use this approach when
the smallest prime divisor of $|G|$ also divides $|H|$.
In the general case, we reduce to the previous case by decomposing $G$ as
a semidirect product.

The algorithmic results are a straightforward generalization of those
of~\cite{DGKS08}. In particular, one merely needs to find the correct way to
generalize the algorithms (replacing the direct product presentation of $G$
with a polycyclic presentation) and verifying that the same analysis goes
through.

%%%%%%%%%%%%%%%%%%%%%%%%%%%%%%%%%%%%%%%%%%%%%%%%%%%%%%%%%%%%%%%%%%%%%%%%%%%%%%%
%%%%%%%%%%%%%%%%%%%%%%%%%%%%%%%%%%%%%%%%%%%%%%%%%%%%%%%%%%%%%%%%%%%%%%%%%%%%%%%
\section{Preliminaries}
\label{section:preliminaries}
%%%%%%%%%%%%%%%%%%%%%%%%%%%%%%%%%%%%%%%%%%%%%%%%%%%%%%%%%%%%%%%%%%%%%%%%%%%%%%%
%%%%%%%%%%%%%%%%%%%%%%%%%%%%%%%%%%%%%%%%%%%%%%%%%%%%%%%%%%%%%%%%%%%%%%%%%%%%%%%

%%%%%%%%%%%%%%%%%%%%%%%%%%%%%%%%%%%%%%%%%%%%%%%%%%%%%%%%%%%%%%%%%%%%%%%%%%%%%%%
\subsection{Group homomorphisms}
%%%%%%%%%%%%%%%%%%%%%%%%%%%%%%%%%%%%%%%%%%%%%%%%%%%%%%%%%%%%%%%%%%%%%%%%%%%%%%%

Let $G$ and $H$ be finite groups, with homomorphisms $\Hom(G,H)$.
A function $\phi:G \to H$ is a (left) affine homomorphism if there exists
$h \in H$ and $\phi_0 \in \Hom(G,H)$ such that $\phi(g) = h\phi_0(g)$ for
every $g \in G$.
We use $\aHom(G,H)$ to denote the set of left affine homomorphisms from $G$ to
$H$.
Note that the set of left affine homomorphisms equals the set of right affine
homomorphisms, since
\[
h\phi_0(g) = (h\phi_0(g)h^{-1})h
\]
and $\psi_0(g) \triangleq h\phi_0(g)h^{-1}$ is a homomorphism.

The \emph{equalizer} of two functions $f,g : G \to H$, denoted
$\Eq(f,g)$, is the subset of $G$ on which $f$ and $g$ agree, i.e.
\[
\Eq(f,g) \triangleq \{x \in G \mid f(x) = g(x)\}.
\]
More generally, if $\Phi \subseteq \{f:G \to H\}$ is a collection of functions,
then the \emph{equalizer} of $\Phi$ is the set
\[
\Eq(\Phi) \triangleq \{x \in G \mid f(x) = g(x)~~\forall f,g \in \Phi\}.
\]
In the theory of error correcting codes, the usual measure of distance between
two strings is the relative Hamming distance, which is the fraction of symbols
on which they differ. In the context of group homomorphisms, we find it
more convenient to study the complementary notion, the fractional agreement.
We define the \emph{agreement} $\agr(f,g)$ between two functions
$f,g:G \to H$ to be the quantity
\[
\agr(f,g) \triangleq \frac{|\Eq(f,g)|}{|G|}.
\]
The \emph{maximum agreement} of the code $\aHom(G,H)$, denoted by
$\Lambda_{G,H}$, is defined as
\[
\Lambda_{G,H} \triangleq
\max_{\substack{\phi,\psi \in \aHom(G,H) \\ \phi \ne \psi}}
\agr(\phi,\psi)
\]

The following theorem gives the value of $\Lambda_{G,H}$ when $G$ is solvable
or $H$ is nilpotent.

\begin{theorem}
[{\cite{G14}}]
\label{theorem:main lambda}
Suppose $G$ and $H$ are finite groups and $G$ is solvable or $H$ is nilpotent.
Then
\[
\Lambda_{G,H} = \frac1p
\]
where $p$ is the smallest prime divisor of $\gcd(|G|,|H|)$ such that
$G$ has a normal subgroup of index~$p$.
If no such $p$ exists, then $|\Hom(G,H)| = 1$; in particular,
$\Lambda_{G,H} = 0$.
\end{theorem}

\begin{proposition}
\label{proposition:product coprime lambda}
If $G$ and $H$ are finite groups and $G = N \rtimes G_1$ for some normal
subgroup $N \triangleleft G$ and subgroup
$G_1 \le G$ and $|\Hom(G_1,H)| = 1$, then every
$\phi \in \aHom(G,H)$ is of the form $\phi(xy) = \psi(x)$ for some
$\psi \in \aHom(N,H)$ and every $x \in N$ and $y \in G_1$.
In particular,
\[
\Lambda_{G,H} \le \Lambda_{N,H}
\]
\end{proposition}

%%%%%%%%%%%%%%%%%%%%%%%%%%%%%%%%%%%%%%%%%%%%%%%%%%%%%%%%%%%%%%%%%%%%%%%%%%%%%%%
\subsection{Some facts about supersolvable groups}
%%%%%%%%%%%%%%%%%%%%%%%%%%%%%%%%%%%%%%%%%%%%%%%%%%%%%%%%%%%%%%%%%%%%%%%%%%%%%%%

\begin{proposition}
\label{proposition:supersolvable normal series sorting}
If $G$ is a finite supersolvable group and $|G| = p_1 \cdots p_k$,
where $p_1 \ge \cdots \ge p_k$ are primes, then $G$ has an normal
cyclic series
\[
\{1_G\} = G_0 \triangleleft G_1 \triangleleft \cdots \triangleleft G_k = G
\]
where each $G_i/G_{i-1} \cong \Z_{p_i}$.
\end{proposition}

The following proposition allows us to decompose a finite supersolvable
group as a semidirect product whose components have coprime orders.

\begin{proposition}
\label{proposition:supersolvable semidirect product}
If $G$ is a finite supersolvable group and
$|G| = p_1^{r_1} \cdots p_m^{r_m}$, where $p_1 > \cdots > p_m$ are prime.
For any $k \in [m]$, $G$ has a normal subgroup $N_k \triangleleft G$
such that $|N_k| = p_1^{r_1} \cdots p_k^{r_k}$,
$|G/N_k| = p_{k+1}^{r_{k+1}} \cdots p_m^{r_m}$, and
$G = N_k \rtimes G/N_k$.
\end{proposition}

%%%%%%%%%%%%%%%%%%%%%%%%%%%%%%%%%%%%%%%%%%%%%%%%%%%%%%%%%%%%%%%%%%%%%%%%%%%%%%%
\subsection{Special intersecting families}
%%%%%%%%%%%%%%%%%%%%%%%%%%%%%%%%%%%%%%%%%%%%%%%%%%%%%%%%%%%%%%%%%%%%%%%%%%%%%%%

\begin{definition}[Special intersecting family]
\label{definition:special intersecting family}
Fix an ambient set $X$. For any subset $S \subseteq X$, define the
\emph{density of $S$ in $X$} to be
\[
\mu(S) = \frac{|S|}{|X|}.
\]
A collection $S_1,\ldots,S_\ell \subseteq X$ of subsets is a
\emph{$(\rho,\tau,c)$-special intersecting family} if the following hold:
\begin{enumerate}
\item%
$\mu(S_i) \ge \rho$ for each $i$;
\item%
$\mu(S_i \cap S_j) \le \rho$ whenever $i \ne j$;
\item%
$\sum_{i=1}^\ell \left(\mu(S_i) - \rho\right)^c \le 1$;
\item%
If $J \subseteq I \subseteq [\ell]$, $|J| \ge 2$, and $\mu(S_I) > \tau$,
then $S_I = S_J$, where $S_K = \cap_{i \in K} S_i$ for any
$K \subseteq [\ell]$;
\end{enumerate}
\end{definition}

For our bounds on the combinatorial list-decodability, we use the same outline
as that of~\cite{DGKS08}. In particular, this involves analyzing the agreement
sets of homomorphisms with the given function and showing that they form a
special intersecting family. The following result of~\cite{DGKS08} allows us
to deduce bounds on the sizes of the agreement sets in terms of the size of
the union.

\begin{theorem}[{\cite[Theorem~3.2]{DGKS08}}]
\label{theorem:special intersecting theorem}
For every $c < \infty$, there exists $C < \infty$ (let's call it $C$ the
\emph{special intersecting number for $c$}) such that the following
holds:
if $S_1,\ldots,S_\ell$ form a $(\rho,\rho^2,c)$-special intersecting family,
with $\mu(S_i) = \rho + \alpha_i$ and $\mu(\cup_i S_i) = \rho + \alpha$, then
\begin{equation}
\alpha^C \ge \sum_{i=1}^\ell \alpha_i^C.
\end{equation}
In fact, one can take $C=2c \cdot (c+1)(4 + (c+1)\log_23)$.
\end{theorem}

We will also use the following $q$-ary Johnson bound (see the appendix
of~\cite{DGKS08} for a proof).

\begin{proposition}[$q$-ary Johnson Bound]
\label{proposition:johnson bound}
Let $f,\phi_1,\ldots,\phi_\ell: [n] \to [q]$ be functions satisfying the
following properties:
\begin{enumerate}
\item%
$\agr(f,\phi_i) = \frac1q + \alpha_i$ for $\alpha_i \ge 0$
\item%
$\agr(\phi_i,\phi_j) \le \frac1q$ for every $i \ne j$.
\end{enumerate}
Then $\sum_{i=1}^\ell \alpha_i^2 \le 1$.
\end{proposition}

%%%%%%%%%%%%%%%%%%%%%%%%%%%%%%%%%%%%%%%%%%%%%%%%%%%%%%%%%%%%%%%%%%%%%%%%%%%%%%%
%%%%%%%%%%%%%%%%%%%%%%%%%%%%%%%%%%%%%%%%%%%%%%%%%%%%%%%%%%%%%%%%%%%%%%%%%%%%%%%
\section{List-decoding radius for supersolvable groups}
\label{section:combinatorial bounds}
%%%%%%%%%%%%%%%%%%%%%%%%%%%%%%%%%%%%%%%%%%%%%%%%%%%%%%%%%%%%%%%%%%%%%%%%%%%%%%%
%%%%%%%%%%%%%%%%%%%%%%%%%%%%%%%%%%%%%%%%%%%%%%%%%%%%%%%%%%%%%%%%%%%%%%%%%%%%%%%

%%%%%%%%%%%%%%%%%%%%%%%%%%%%%%%%%%%%%%%%%%%%%%%%%%%%%%%%%%%%%%%%%%%%%%%%%%%%%%%
\subsection{Preliminary notation and definitions}
%%%%%%%%%%%%%%%%%%%%%%%%%%%%%%%%%%%%%%%%%%%%%%%%%%%%%%%%%%%%%%%%%%%%%%%%%%%%%%%

If $H$ is supersolvable, we may write
\[
H = H_0 \triangleright H_1 \triangleright \cdots \triangleright H_m = \{1\}
\]
where $H_{i-1}/H_i \cong \Z_{p_i}$.
For $k \in [m]$, define $H_{(k)} \triangleq H/H_k$, which
is a group since $H_k$ is normal in $H$.
In particular, $H_{(0)} = \{1\}$ and $H_{(m)} = H$.

Given $f: G \to H$ and $k \in [m]$, define $f^{(k)}:G \to H_{(k)}$
and $f^{(-k)}: G \to H_k$ as follows.
Define $f^{(k)}: G \to H_{(k)}$ to be $f$ composed with the natural
quotient map, sending $x \in G$ to the coset $f(x)H_k$ of $H_k$.
Therefore, $f^{(k)}$ is an (affine) homomorphism if $f$ is.
To define the latter map, we need
to choose, for each $i \in [0,m-1]$, an element $y_i \in H_i \setminus
H_{i+1}$. Then each $k$-tuple $(a_0,\ldots,a_{k-1})$, where
$0 \le a_j \le p-1$, corresponds to a distinct coset
$y_0^{a_0} \cdots y_{k-1}^{a_{k-1}}H_k$.
If $f(x)H_k = y_0^{a_0} \cdots y_{k-1}^{a_{k-1}}H_k$, then define
$f^{(-k)}(x) \triangleq y_{k-1}^{-a_{k-1}} \cdots y_0^{-a_0} f(x)$.
Note that $f^{(-k)}(x) \in H_k$ but $f^{(-k)}$ may not be a homomorphism
in general (even if $f$ is).
Also, note that $f$ is determined by $f^{(k)}$ and $f^{(-k)}$:
if $f^{(k)}(x) = y_0^{a_0} \cdots y_{k-1}^{a_{k-1}}H_k$, then
$f(x) = y_0^{a_0} \cdots y_{k-1}^{a_{k-1}}f^{(-k)}(x)$.

If $i < j$ and $\phi: G \to H_{(i)}$ and $\psi: G \to H_{(j)}$, then
$\psi$ \emph{extends} $\phi$ if $\psi^{(i)} = \phi$. Here,
$\psi^{(i)}$ makes sense, because $H_j < H_i$, and so we get a chain
$H_0/H_j \triangleright H_1/H_j \triangleright \cdots \triangleright H_j/H_j
= \{1\}$ induced by the original chain for $H$, and so $\psi^{(i)}$ is
just $\psi$ composed by modding out by $H_i/H_j$. One can then define
$\psi^{(-i)}$ to make sense too.

%%%%%%%%%%%%%%%%%%%%%%%%%%%%%%%%%%%%%%%%%%%%%%%%%%%%%%%%%%%%%%%%%%%%%%%%%%%%%%%
\subsection{Combinatorial bounds for agreement $\Lambda_{G,H} + \epsilon$}
%%%%%%%%%%%%%%%%%%%%%%%%%%%%%%%%%%%%%%%%%%%%%%%%%%%%%%%%%%%%%%%%%%%%%%%%%%%%%%%

We begin with the case where the smallest prime divisor of $|G|$ also divides
$|H|$.
\begin{theorem}
\label{theorem:supersolvable G-friendly combinatorial}
There exists a universal constant $C < \infty$ such that whenever $G$ and $H$
are finite supersolvable groups and $H$ is $G$-friendly, i.e.\ the smallest
prime divisor $p$ of $|G|$ also divides $|H|$,
then for any $f: G \to H$ and $\epsilon > 0$,
there are at most $(1/\epsilon)^C$ affine homomorphisms $\phi \in \aHom(G,H)$
such that $\agr(\phi,f) \ge \frac1p + \epsilon$.
\end{theorem}

\begin{proof}
Let $p_1 \le \cdots \le p_m$ be primes such that $|H| = p_1 \cdots p_m$.
By Proposition~\ref{proposition:supersolvable normal series sorting},
$H$ has an normal cyclic series
\[
H = H_0 \triangleright H_1 \triangleright \cdots \triangleright H_m = \{1_H\}
\]
where $H_{i-1}/H_i \cong \Z_{p_i}$ for each $i$.

\begin{claim}
For $k \in [0,m]$, if $\phi \in \aHom(G,H_{(k)})$ satisfies $\agr(\phi,f^{(k)})
= \frac1p + \alpha$ for some $\alpha \ge \epsilon$, then the number of
$\psi \in \aHom(G,H)$ extending $\phi$ with
$\agr(\psi,f) \ge \frac1p + \epsilon$ is at most $(\alpha/\epsilon)^C$.
\end{claim}
\begin{proof}
We induct backwards on $k$. The base case $k=m$ is trivial. Now suppose
$k < m$ and the claim holds for $k+1$.
Let $\phi_1,\ldots,\phi_\ell \in \aHom(G,H_{(k+1)})$ be all the homomorphisms
extending $\phi$ with $\agr(\phi_i,f^{(k+1)}) \ge \frac1p + \epsilon$.
Define $\alpha_i \triangleq \agr(\phi_i,f^{(k+1)}) - \frac1p$.
Define $S_i \triangleq \Eq(\phi_i,f^{(k+1)})$.
We claim that $S_1,\ldots,S_\ell$ form a
$\left(\frac{1}{p},\frac{1}{p^2},2\right)$-special intersecting family.
Before we prove this, we show how it implies the claim.
By Theorem~\ref{theorem:special intersecting theorem},
$(\alpha')^C \ge \sum_{i=1}^\ell \alpha_i^C$, where
$\alpha' = \mu(\cup_iS_i) - \frac1p$. But $\cup_i S_i \subseteq \Eq(\phi,f)$,
so $\alpha \ge \alpha'$, and thus $\alpha^C \ge \sum_{i=1}^\ell \alpha_i^C$.
Moreover, every $\psi \in \aHom(G,H)$ extending $\phi$ with
$\agr(\psi,f) \ge \frac1p+\epsilon$ must extend one of the $\phi_i$.
By induction, there are at most $(\alpha_i/\epsilon)^C$ such $\psi$
extending $\phi_i$. Hence, there are at most
$\sum_{i=1}^\ell (\alpha_i/\epsilon)^C \le (\alpha/\epsilon)^C$
such $\psi$ extending $\phi$.

Now, we show that $S_1,\ldots,S_\ell$ form a
$\left(\frac{1}{p},\frac{1}{p^2},2\right)$-special intersecting family.
We verify the four properties:
\begin{enumerate}[(1)]
\item%
By definition, we have $\mu(S_i) = \frac1p + \alpha_i \ge \frac1p$.

\item%
If $i \ne j$, then since $\phi_i,\phi_j \in \aHom(G,H_{(k+1)})$, we have
$S_i \cap S_j \subseteq \Eq(\phi_i,\phi_j)$ and therefore
$\mu(S_i \cap S_j) \le \agr(\phi_i,\phi_j) \le \Lambda_{G,H_{(k+1)}}
\le \Lambda_{G,H} \le \frac1p$.

\item%
Define $g \triangleq (f^{(k+1)})^{(-k)}:G \to H_k/H_{k+1}\cong\Z_{p_{k+1}}$ and
define $\psi_i \triangleq \phi^{(-k)}:G \to H_k/H_{k+1} \cong \Z_{p_{k+1}}$.
If $\phi_i(x) = f^{(k+1)}(x)$, then $\psi_i(x) = g(x)$, so certainly
$\agr(g,\psi_i) \ge \agr(f^{(k+1)},\phi_i) = \frac1p + \alpha_i$.
Moreover, if $i \ne j$, since $\phi_i,\phi_j$ both extend $\phi$, then
$\phi_i(x) = \phi_j(x)$ if and only if $\psi_i(x) = \psi_j(x)$, so
$\agr(\psi_i,\psi_j) = \agr(\phi_i,\phi_j) \le \Lambda_{G,H_{(k+1)}}
\le \Lambda_{G,H} = \frac1p$.

\item%
Suppose $J \subseteq I$, $|J| \ge 2$, and $\mu(S_I) > 1/p^2$.
Define $\Phi_I \triangleq \{\phi_i \mid i \in I\}$ and define $\Phi_J$
similarly. Then $S_I \subseteq \Eq(\Phi_I)$ and $S_J \subseteq \Eq(\Phi_J)$,
and since $|J| \ge 2$, we have $1/p^2 < \mu(\Eq(\Phi_I)) \le \mu(\Eq(\Phi_J))
\le 1/p$. But $\mu(\Eq(\Phi_J))/\mu(\Eq(\Phi_I))$ divides $|G|$ and
$p$ is the smallest prime divisor of $|G|$, so it must be that
$\mu(\Eq(\Phi_I)) = \mu(\Eq(\Phi_J))$, and hence $\Eq(\Phi_I) = \Eq(\Phi_J)$.
Fix any $j \in J$.
Then $S_I = S_j \cap \Eq(\Phi_I) = S_j \cap \Eq(\Phi_J) = S_J$.

\end{enumerate}
\end{proof}
The theorem follows by taking $k=0$ in the claim.
\end{proof}

Before we prove the general case, we first prove a useful lemma.
In what follows, for any code $\cC \subseteq \Sigma^n$ and agreement parameter
$a \in [0,1]$, define $\ell(\cC,a)$ to be the quantity
\[
\ell(\cC,a) \triangleq \max_{w \in \Sigma^n}
|\{c \in \cC \mid \agr(c,w) \ge a \}|.
\]

\begin{lemma}
\label{lemma:repetition}
Let $\cC \subseteq \Sigma^n$ be a code.
If $s > r \ge 1$, and
$\cC_r \triangleq \{(\underbrace{c,\ldots,c}_r) \in \Sigma^{rn}\mid c\in\cC\}$ and
$\cC_s \triangleq \{(\underbrace{(c,\ldots,c)}_s) \in\Sigma^{sn}\mid c\in\cC\}$, then for
any $a \in [0,1]$,
\[
\ell(\cC_r,a) \le \ell(\cC_s,\lfloor s/r \rfloor (r/s) \cdot a).
\]
\end{lemma}
\begin{proof}
Let $w \in \Sigma^{rn}$ such that
$|\{\underbrace{(c,\ldots,c)}_r \in \cC_r \mid
\agr((\underbrace{c,\ldots,c}_r),w) \ge a\}|
= \ell(\cC_r,a)$. Define $w' \in \Sigma^{sn}$ by
$w' = (\underbrace{w,\ldots,w}_{\lfloor s/r \rfloor},w'')$,
where $w'' \in \Sigma^{(s-\lfloor s/r \rfloor r)n}$ is defined arbitrarily. Then for each $c \in \cC$ such that
$\agr((\underbrace{c,\ldots,c}_r),w) \ge a$,
\begin{eqnarray*}
\agr((\underbrace{c,\ldots,c}_s),w')
&\ge& \frac{1}{sn} \left(\lfloor s/r \rfloor \cdot
rn \cdot \agr((\underbrace{c,\ldots,c}_r),w)\right) \\
&\ge& \left\lfloor \frac{s}{r} \right\rfloor \frac{r}{s} \cdot a.
\end{eqnarray*}
\end{proof}

\begin{theorem}
\label{theorem:supersolvable combinatorial}
There exists a universal constant $C < \infty$ such that whenever
$G$ and $H$ are finite supersolvable groups, then for any $f:G \to H$
and $\epsilon > 0$, there are at most $(1/\epsilon)^C$ affine homomorphisms
$\phi \in \aHom(G,H)$ such that $\agr(\phi,f) \ge \Lambda_{G,H} + \epsilon$.
\end{theorem}
\begin{proof}
Let $p$ be the smallest prime divisor of $\gcd(|G|,|H|)$ such that
$G$ has a normal subgroup of index $p$, so that
$\Lambda_{G,H} = \frac1p$ (Theorem~\ref{theorem:main lambda}).
If $p$ is the smallest prime divisor of $|G|$, then
the result follows from
Theorem~\ref{theorem:supersolvable G-friendly combinatorial},
so suppose $p$ is not the smallest prime divisor of $|G|$.
By Proposition~\ref{proposition:supersolvable semidirect product},
we can write $G = N \rtimes G'$ for some proper normal subgroup
$N \triangleq G$ where $p$ is the smallest prime divisor
of $|N|$ and every prime dividing $|G'|$ is smaller than $p$, and
therefore $\gcd(|G'|,|H|) = 1$.
By Proposition~\ref{proposition:product coprime lambda},
every $\phi \in \aHom(G,H)$ is of the form $\phi(x,y) = \psi(x)$ for
$x \in N$ and $y \in G'$.
Thus, $\aHom(G,H)$ is isomorphic to the code
\[
\cC_r \triangleq \{(\underbrace{\psi,\ldots,\psi}_r) \mid
\psi \in \cC \}
\]
where $\cC = \aHom(N,H)$ and $r = |G'|$. Let $q > \max\{|G|,|H|\}$ be a prime
and consider the group $G'' \triangleq N \times \Z_q$, which is supersolvable.
Then
$\aHom(G'',H)$ is isomorphic to the code
\[
\cC_q \triangleq \{(\underbrace{\psi,\ldots,\psi}_{q}) \mid \psi \in \cC\}.
\]
Letting $a \triangleq \frac1p + \epsilon$, applying
Lemma~\ref{lemma:repetition} and
Theorem~\ref{theorem:supersolvable G-friendly combinatorial}
(using the fact that $H$ is $G''$-friendly), we get an upper bound of
\[
\left(
\frac{1}{(\lfloor q/|G'| \rfloor (|G'|/q) - 1)\frac1p +
\lfloor q/|G'| \rfloor (|G'|/q) \cdot \epsilon}
\right)^C
\le
\left(
\frac{1}{\left(1 - \frac{|G'|}{q} \right)\epsilon -
\frac{|G'|}{q}\frac1p}
\right)^C
\]
affine homomorphisms $\phi \in \aHom(G,H)$ with $\agr(\phi,f) \ge
\frac1p + \epsilon$. By taking $q \to \infty$, the above upper bound
approaches $(1/\epsilon)^C$.
\end{proof}

%%%%%%%%%%%%%%%%%%%%%%%%%%%%%%%%%%%%%%%%%%%%%%%%%%%%%%%%%%%%%%%%%%%%%%%%%%%%%%%
\subsection{Exponential list size for agreement $\Lambda_{G,H}$}
%%%%%%%%%%%%%%%%%%%%%%%%%%%%%%%%%%%%%%%%%%%%%%%%%%%%%%%%%%%%%%%%%%%%%%%%%%%%%%%

We conclude this section by showing that if $G$ is solvable, then the list
size for agreement $\Lambda_{G,H}$ can be exponential in
$\log|G| + \log|H|$, showing that the list-decoding distance we achieve is
optimal. In other words, we have identified the list-decoding radius for
$\aHom(G,H)$ when $G$ and $H$ are supersolvable.

In fact, we observe that the list size can be $\Omega(|G| \cdot |H|)$ even just
for abelian $G$ and $H$, when $\Lambda_{G,H} = \frac1p$ is fixed.
Let $G = \Z_p^n$ and $H = \Z_p^m$, so that $\Lambda_{G,H} = \frac1p$.
Consider the maps $\phi_{a,b}$, where $a \in \Z_P^n$ and $b \in \Z_p^m$
are nonzero vectors, defined by
\[
\phi_{a,b}(x_1,\ldots,x_n) = (a_1x_1 + \cdots + a_nx_n)b.
\]
Note that $\agr(\phi_{a,b}, 0) = \frac1p$. Moreover, there are
$p^n-1$ choices for $a$ and $p^m-1$ choices for $b$, and $\phi_{a,b} =
\phi_{c,d}$ if and only if there exists $\lambda \in \Z_p^*$ such that
$c = \lambda a$ and $b = \lambda d$. So the number of distinct homomorphisms
agreeing with the zero function is
$\frac{(p^n-1)(p^m-1)}{p-1} = \Omega(|G|\cdot |H|) = \exp(\log|G| + \log|H|)$.

%%%%%%%%%%%%%%%%%%%%%%%%%%%%%%%%%%%%%%%%%%%%%%%%%%%%%%%%%%%%%%%%%%%%%%%%%%%%%%%
%%%%%%%%%%%%%%%%%%%%%%%%%%%%%%%%%%%%%%%%%%%%%%%%%%%%%%%%%%%%%%%%%%%%%%%%%%%%%%%
\section{Algorithm for supersolvable $G$}
\label{section:algorithm}
%%%%%%%%%%%%%%%%%%%%%%%%%%%%%%%%%%%%%%%%%%%%%%%%%%%%%%%%%%%%%%%%%%%%%%%%%%%%%%%
%%%%%%%%%%%%%%%%%%%%%%%%%%%%%%%%%%%%%%%%%%%%%%%%%%%%%%%%%%%%%%%%%%%%%%%%%%%%%%%

In this section we prove the following theorem.

\begin{theorem}
\label{theorem:list decoding algorithm}
There exists an algorithm $\cA$ such that for every pair of finite
groups $G,H$ where $G$ is solvable and $H$ is supersolvable,
and every $\epsilon > 0$,
$\cA$ is a
$(\Lambda_{G,H}+\epsilon,\poly(\log|G|,\log|H|,\frac{1}{\epsilon}))$-local
list decoder for $\aHom(G,H)$,
provided that $\cA$ has oracle access to the multiplication table of $H$.
\end{theorem}

\subsection{Algorithm}
Let
\[
G = G_k \triangleright G_{k-1} \triangleright \cdots \triangleright
G_0 = \{1_G\}
\]
be a subnormal cyclic series, with $G_i/G_{i-1} \cong \Z_{p_i}$,
$p_1 \ge p_2 \ge \cdots \ge p_k$ and
representatives $g_i \in G_i \setminus G_{i-1}$.
Our main algorithm is Algorithm~\ref{algorithm:listdecode}, which uses
Algorithms~\ref{algorithm:extend} and~\ref{algorithm:prune} as subroutines.

\begin{algorithm}
\caption{List decode}
\label{algorithm:listdecode}
\begin{algorithmic}
\Procedure{ListDecode}{$f$,$G$,$H$}
	\State $\cL \gets \emptyset$
	\Repeat
		\State $S_0 \gets \emptyset$
		\For{$i=1$ to $k$}
			\State $S'_i \gets \textsc{Extend}(i,S_{i-1})$
			\State $S_i \gets \textsc{Prune}(i,S'_i)$
		\EndFor
		\ForAll{$\phi \in S_k$}
			\State $B \gets \textsc{FrequentValues}
				(x \mapsto f(x)\phi(x)^{-1}, \Lambda_{G,H} + \epsilon/2)$
			\State $\cL \gets \cL \cup \{x \mapsto b\phi(x)\mid b \in B\}$
		\EndFor
	\Until{$C\log\frac{1}{\epsilon}$ times}
\EndProcedure
\end{algorithmic}
\end{algorithm}

\begin{algorithm}
\caption{Extend}
\label{algorithm:extend}
\begin{algorithmic}
\Procedure{Extend}{$i$,$S$}
	\State $S' \gets \emptyset$
	\ForAll{$\phi \in S$}
		\Repeat
			\State Pick $(\alpha_{i+1},\ldots,\alpha_k) \in
			\Z_{p_{i+1}} \times \cdots \times \Z_{p_k}$ uniformly at random
			\State $s \gets g_k^{\alpha_k} \cdots g_{i+1}^{\alpha_{i+1}}$
			\State Pick $y_1,y_2 \in G_{i-1}$ and $c_1,c_2 \in \Z_{p_i}$
			uniformly at random
			\If{$c_1 - c_2$ is invertible modulo $p_1 \cdots p_i$}
				\State $\gamma \gets (c_1 - c_2)^{-1} \in
				\Z^\times_{p_1 \cdots p_i}$
				\State $a \gets
				(\phi(y_2)f(sg_i^{c_2}y_2)^{-1}f(sg_i^{c_1}y_1)\phi(y_1)^{-1}
				)^{\gamma}$
				\State Define $\theta:G_i \to H$ by
				$\theta(g_i^cx) = a^c\phi(x)$
				\State $S' \gets S' \cup \{\theta\}$
			\EndIf
		\Until{$(\log|G|\log|H|\frac{1}{\epsilon})^4$ times}
	\EndFor
	\State \Return $S'$
\EndProcedure
\end{algorithmic}
\end{algorithm}

\begin{algorithm}
\caption{Prune}
\label{algorithm:prune}
\begin{algorithmic}
\Procedure{Prune}{$i$,$S$}
	\State $S' \gets \emptyset$
	\Repeat
		\State Pick $(\alpha_{i+1},\ldots,\alpha_k) \in
		\Z_{p_{i+1}} \times \cdots \times \Z_{p_k}$ uniformly at random
		\State $s \gets g_k^{\alpha_k} \cdots g_{i+1}^{\alpha_{i+1}}$
		\ForAll{$\phi \in S$}
			\State $B \gets \textsc{FrequentValues}
			(x \mapsto f(sx)\phi(sx)^{-1}, \Lambda_{G,H} + \epsilon/2)$
			\If{$|B| \ge 1$}
				\State $S' \gets S' \cup \{\phi\}$
			\EndIf
		\EndFor
	\Until{$(\log|G|\log|H|\frac{1}{\epsilon})^2$ times}
	\If{$|S'| > (\log|G|\log|H|\frac{1}{\epsilon})^{2C}$}
		\State \Return error
	\EndIf
	\State \Return $S'$
\EndProcedure
\end{algorithmic}
\end{algorithm}

The analysis is the same as in~\cite{DGKS08}.

\newpage

%%%%%%%%%%%%%%%%%%%%%%%%%%%%%%%%%%%%%%%%%%%%%%%%%%%%%%%%%%%%%%%%%%%%%%%%%%%%%%%
\bibliographystyle{alpha}
\bibliography{homomorphism}
%%%%%%%%%%%%%%%%%%%%%%%%%%%%%%%%%%%%%%%%%%%%%%%%%%%%%%%%%%%%%%%%%%%%%%%%%%%%%%%

%%%%%%%%%%%%%%%%%%%%%%%%%%%%%%%%%%%%%%%%%%%%%%%%%%%%%%%%%%%%%%%%%%%%%%%%%%%%%%%
\end{document}